\documentclass[runningheads]{llncs}
\usepackage[utf8]{inputenc}
\usepackage[english]{babel}
\usepackage[T1]{fontenc}
\usepackage{amsmath,amsthm}
\usepackage{amsfonts}
\usepackage{mathptmx}
\usepackage{amssymb}
\usepackage{float}
\usepackage{hyperref}
\usepackage{comment}
\usepackage{graphicx}
\usepackage{tikz}
\usepackage{mathtools}

\newtheorem{teo}{Theorem} 

\newtheorem{lem}[teo]{Lemma}
\newtheorem{pro}[teo]{Proposition}
\newtheorem{defe}[teo]{Definition}
\newtheorem{ex}[teo]{Example}

\newtheorem{rem}[teo]{Remark}
\newcommand{\resp}[1]{\ (resp. #1)}

\usepackage{xifthen}
\newcommand{\ifnv}[2]{\ifthenelse{\equal{#1}{}}{}{#2}}
\newcommand\sett[3][]{\left\{\left.#2\ifnv{#1}{\in #1}\vphantom{#3}\right|#3\right\}}

\newcommand{\len}[2][]{\left|#2\right|_{#1}}
\newcommand{\abs}[1]{\left|\vphantom{f^f_f}#1\right|}

\newcommand\dc{\mathfrak d_C}
\newcommand{\dll}{{d_L}}
\newcommand{\dhh}{d_{H}}
\newcommand\diam\theta
\newcommand\maxf[1]{\left\|{#1}\right\|}
\newcommand\minf[1]{\left|{#1}\right|}

\newcommand{\uinf}[1]{#1^\infty}

\newcommand{\N}{\mathbb N}
\newcommand{\Ns}{\mathbb N\setminus\{0\}}

\usepackage{stmaryrd}
\newcommand{\co}[2]{\left\llbracket #1,#2\right\llbracket}
\newcommand{\oc}[2]{\left\rrbracket #1,#2\right\rrbracket}

\newcommand\maxd[1][f]{d_{#1}^+}
\newcommand\mind[1][f]{d_{#1}^-}
\newcommand\defeq{:=}
\newcommand{\ie}{\textit{i.e.,}\ }

\newcommand\besi[1]{\mathfrak d_{#1}}
\newcommand\weyl[1]{\mathfrak{\hat d}_{#1}}

\newcommand{\ipart}[1]{\left\lfloor #1\right\rfloor}
\newcommand{\spart}[1]{\left\lceil #1\right\rceil}

\usepackage{fancyhdr}
\begin{document}
\date{}
\title{Dill maps in the Weyl-like space associated to the Levenshtein distance.}
\author{Firas BEN RAMDHANE\textsuperscript{1} and Pierre GUILLON\textsuperscript{2}}
\institute{\begin{center}
	\textsuperscript{1,2}Aix Marseille Univ, CNRS, Centrale Marseille, I2M, Marseille, France, \\
	\textsuperscript{1}Sfax University, Faculty of Sciences of Sfax, Tunisia.\\
	\textsuperscript{1}\email{firasbenramdhane@math.cnrs.fr}. \\
	\textsuperscript{2}	\email{pguillon@math.cnrs.fr}.\\
\end{center}}
\maketitle

\begin{abstract}
The Weyl pseudo-metric is a shift-invariant pseudo-metric over the set of infinite sequences, that enjoys interesting properties and is suitable for studying the dynamics of cellular automata. It corresponds to the asymptotic behavior of the Hamming distance on longer and longer subwords.
In this paper we characterize well-defined dill maps (which are a generalization of cellular automata and substitutions) in the Weyl space and the sliding Feldman-Katok space where the Hamming distance appearing in the Weyl pseudo-metrics is replaced by the Levenshtein distance.

\end{abstract}

%
\section{Basic definitions and notations}\label{s:def}
\paragraph{\textbf{Word combinatorics.}}
We fix once and for all an \emph{alphabet} $A$ of finitely many \emph{letters}.
A \emph{finite word} over $A$ is a finite sequence of letters in $A$; it is convenient to write a word as $u=u_{\co{0}{\len u}}$ to express $u$ as the concatenation of the letters $u_0,u_1,\ldots,u_{\len{u}-1}$, with $\len{u}$ representing the length of $u$, \ie the number of letters appearing in $u$, and $\co{0}{\len u}=\{0,...,|u|-1\}$.
The unique word of length $0$ is the empty word denoted by $\lambda$.
A \emph{configuration}  $x=x_0x_1x_2\ldots$ over $A$ is the concatenation of infinitely many letters from $A$.
The set of all finite \resp{infinite} words over $A$ is denoted by $A^*$ \resp{$A^\N$}, $A^n$ is the set of words of length $n\in \N$ and $A^+=A^*\setminus\{\lambda\}$.
\paragraph{\textbf{Topologies over $A^\N$.}}
Most classically, the set $A^\N$ is endowed with the product topology of the discrete topology on each copy of $A$. The topology defined on $A^\N$ is metrizable, corresponding to the \emph{Cantor distance} denoted by $\dc$ and defined as follows:
$$ \dc(x,y)=2^{-\min\sett{n\in \N}{x_n\neq y_n}}, \forall x\neq y\in A^\N, \text{ and } \dc(x,x)=0, \forall x\in A^\N.$$



Topological dynamical systems were studied using other topologies on the set of infinite words, such as the Besicovitch and Weyl spaces \cite{blanchard_cellular_1997} and the Feldman-Katok space \cite{ramdhane2022cellular}. The Weyl space and a similar space that are defined using pseudo-metrics depend on the two following distances are of interest to us in this paper.
\begin{defe}
\begin{enumerate}
\item The \emph{Hamming distance} denoted by $d_H$ and defined over finite words of the same length $u,v$ by: 
$d_H(u,v)=|\sett{i\in\co{0}{|u|}}{u_i\neq v_i}|.$
\item The \emph{Levenshtein distance} $\dll$ is defined over $u,v\in A^*$ as follows:
		\[\dll(u,v)=\frac{1}{2}\min\sett{m+m'}{\exists j_1<\cdots <j_m, j'_1<\cdots <j'_{m'}, D_{j_1}\circ\ldots \circ D_{j_m}(u)=D_{j'_1}\circ\ldots \circ D_{j'_{m'}}(v)},\]
		where  $D_j$ is \emph{deletion} operation at position $j\in \co{0}{\len{u}}$ is defined over word $u\in A^*$ as follows: $D_j(u)= u_0u_1\ldots u_{j-1}u_{j+1}\ldots u_{\len{u}-1}.$	
\end{enumerate}
\end{defe}
\begin{defe}
\begin{enumerate}
\item The Weyl pseudo-metric, denoted by $\weyl{H}$, is defined as follows:
	$$\weyl{H}(x,y)=\limsup_{\ell\to \infty}\max_{k\in \N} \frac{d_H(x_{\co{k}{k+\ell}},y_{\co{k}{k+\ell}})}{\ell}, \forall x,y\in A^\N,$$
\item 	The sliding pseudo-metric associated to the Levenshtein distance, denoted by $\weyl{L}$, is defined as follows:
	$$\weyl{L}(x,y)=\limsup_{\ell\to \infty}\max_{k\in \N} \frac{d_L(x_{\co{k}{k+\ell}},y_{\co{k}{k+\ell}})}{\ell}, \forall x,y\in A^\N.$$
\end{enumerate}
\end{defe}
It is easy to verify that these are pseudo-metrics \ie symmetric, zero over diagonal pairs, and satisfies the triangular inequality.
On the other hand, these are not distances since we can find two different configurations between which the pseudometric is worth zero (for example, we can take two configurations with finitely many of differences).
Hence, it is relevant to quotient the space of configurations by the equivalence of zero distance, in order to get a separated topological space:
\begin{defe}
The relation  $x \sim_{\weyl{*}} y \iff \weyl{*}(x,y)=0$, is an equivalence relation. The quotient space $A^\N /\sim_{\weyl{*}}$ called the Weyl space for $\weyl{*}=\weyl{H}$ and the sliding Feldman-Katok space when $\weyl{*}=\weyl{L}$, denoted $X_{\weyl{*}}$, where $\weyl{*}$ represent previous pseudo-metrics. We denote by $x_{\weyl{*}}$ the equivalence class of $x\in A^\N$ in the quotient space.
Any map $F$ of $A^\N$ to itself such that $\weyl{*}(x,y)=0\implies\weyl{*}(F(x),F(y))=0$ for all $x,y\in A^\N$, induces a well-defined map $F_{\weyl{*}}: X_{\weyl{*}}\to X_{\weyl{*}}$ over the quotient space.
A map $F:A^\N\mapsto A^\N$ is $\weyl{*}$-constant if for all $x,y\in A^\N$, $\weyl{*}(F(x), F(y))=0$.
\end{defe}
\paragraph{\textbf{Dill maps.}}
Dill maps were defined in \cite[Definition~2]{salo2015block}, and generelize both substitutions \cite{fogg2002substitutions} and cellular automata \cite{cant}. 
Here we give an equivalent definition.
\begin{defe}~\begin{enumerate}
\item A map $F: A^\N\mapsto A^\N$ is a \emph{dill map} if there exist a diameter $\diam\in \Ns$ and a \emph{local rule} $f~:~A^{\diam}\to A^+$ such that for all $x,y\in A^\N$:
$F(x)=f(x_{\co{0}{\diam}})f(x_{\co{1}{\diam+1}})f(x_{\co{2}{\diam+2}})\cdots.$
\item The \emph{lower norm} $\minf{f}$ and the \emph{upper norm} $\maxf{f }$ of a dill map $F $ with diameter $\diam$ and local rule $f$ are defined by: 
$\minf f=\min\sett{\len{f(u)}}{u\in A^\diam} \text{ and } \maxf f=\max\sett{\len{f(u)}}{u\in A^\diam}.$
\item
We extend the local rule into a self-map $f^*:A^*\to A^*$ by: 
$f^*(u)=
f(u_{\co{0}{\diam}})f(u_{\co{1}{1+\diam}})\ldots f(u_{\co{\len{u}-\diam}{\len u}}),$
for $u$ such that $\len u\ge \diam$ and $f^*(u)=\lambda$ if $\len u<\diam$.
\item
If $\maxf f=\minf f$, then we say that $F$ is \emph{uniform}.
\end{enumerate}
\end{defe}
When it is clear from the context, we may identify a dill map with its local rule.
\begin{rem}~
\begin{enumerate}
\item Uniform dill maps with $\minf f=\maxf f=1$ are called cellular automata. 
\item The local rule of a dill maps with diameter $\diam=1$ is called substitution. In this case, we denote $\tau$ for the local rule and $\bar\tau$ for the dill map.
\item The composition of a substitution $\tau$ and a cellular automaton local rule $f$ with diameter $\diam$ is a dill map local rule $\tau\circ f$ with diameter $\diam$.
\end{enumerate}
\end{rem}
\begin{ex}
\begin{enumerate}
\item The shift is the CA with diameter $\diam=2$ and local rule $f$ defined by $f(u_0u_1)=u_1$ for all $u_0,u_1\in A$.
\item Let $A=\{a,b\}$. The \emph{Xor} is the CA with diameter $\diam=2$ and local rule $f$ defined by:
$f(aa)=f(bb)=a$ and $f(ab)=f(ba)=b$.
\item The \emph{Fibonacci} substitution defined over $A$ by: 
$\tau(a)=ab \text{ and } \tau(b)=a.$
\item Let $f$ be the local rule of the Xor CA and $\tau$ be the Fibonacci substitution. Then $\tau\circ f$ is a local rule of a dill map with diameter $2$ and defined as follows:
\[\tau\circ f(aa)=\tau\circ f(bb)=ab \text{ and } \tau\circ f(ba)=\tau\circ f(ab)=a.\]
\end{enumerate}
\end{ex}
In the Cantor space, an elegant characterization of cellular automata was given by Curtis, Hedlund and Lyndon in \cite{hedlund1969endomorphisms} as follows: 
A function $F:A^\N\to A^\N$ is a cellular automaton if and only if it is continuous with respect to the Cantor metric and shift-equivariant (\ie $F(\sigma(x))=\sigma(F(x))$, for all $x\in A^\N$).
Similarly to the case of cellular automata, we gave a characterization of dill maps à la Hedlund.
Recall that $\N$ can be naturally endowed with the discrete topology.
\begin{teo}[{\cite[Theorem 11]{ramdhane2022cellular}}]\label{dill cant}
A function $F:A^\N\to A^\N$ is a dill map if and only if it is continuous over the Cantor space and there exists a continuous map $s:A^\N\to \N$ such that for all $x\in A^\N$: $F(\sigma(x))=\sigma^{s(x)}(F(x)).$
\end{teo}

Before proving our main results, let us mention that this paper is a continuation of our previous work \cite{ramdhane2022cellular}, and we suggest that the reader check it out.
\section{Lipschitz property of dill maps with respect to $\weyl{H}$}
It is known since \cite{blanchard_cellular_1997} that every cellular automaton induces a (well-defined) Lipschitz function over the Weyl space. 
Some dill maps, on the contrary, are not well-defined.
\begin{ex}\label{x:fibobes}
	The Fibonacci substitution is not well-defined over the Weyl space $X_{\weyl{H}}$.
	For example, $\weyl{H}(\uinf a,b\uinf a)=0$ but $\weyl{H}(\overline{\tau}(\uinf a),\overline{\tau}(b\uinf a))=\weyl{H}(\uinf{(ab)},\uinf{(ba)})=1$.
\end{ex}
Let us denote, for a uniform dill map $F$ with local rule $f$ and diameter $\diam$:
$$ \maxd=\max\sett{\dhh(f(u),f(v))}{u,v\in A^\diam} \text{ and } \mind=\min\sett{\dhh(f(u),f(v))}{u\ne v\in A^\diam}.$$ 
\begin{lem}\label{lip:lemma}
Let $F$ be a uniform dill map with diameter $\diam$ and local rule $f$. Then for all $\ell,k\in \N$, for $m=\spart{\frac{k}{\minf f}}$, and for $p=\ipart{\frac{\ell+k}{\minf f}}-(m+1)$: 
$$\dhh(F(x)_{\co{k}{k+\ell}},F(y)_{\co{k}{k+\ell}})\le \dhh(x_{\co{m}{m+p+\diam}},y_{\co{m}{m+p+\diam}})\diam\maxd+2\minf f, \forall x,y\in A^\N.$$
\end{lem}
\begin{proof}
Let $x,y\in A^\N$ and $\ell,k\in \N$. Since $F$ is uniform, we can write: 
$$F(x)_{\co{k}{k+\ell}}= vf^*(x_{\co{m}{m+p+\diam}})w, \text{ and }  F(y)_{\co{k}{k+\ell}}= v'f^*(y_{\co{m}{m+p+\diam}})w',$$
where $|v|=|v'|\le\minf f$ and $|w|=|w'|\le\minf f$.
By additivity, we can write then: 
\begin{align*}
  \dhh(F(x)_{\co{k}{k+\ell}},F(y)_{\co{k}{k+\ell}})-2\minf f &\le \sum_{i=0}^{p} \dhh(f(x_{\co{m+i}{m+i+\diam}}),f(y_{\co{m+i}{m+i+\diam}}))\\
  &\le \sum_{\begin{subarray}{c} i=0 \\ x_{\co{m+i}{m+i+\diam}}\neq y_{\co{m+i}{m+i+\diam}}\end{subarray}}^{p} \dhh(f(x_{\co{m+i}{m+i+\diam}}),f(y_{\co{m+i}{m+i+\diam}})) \\ 
& \le \sum_{\begin{subarray}{c} i=0 \\ \exists j\in \co{m+i}{m+i+\diam},  x_j\neq y_j\end{subarray}}^{p} \maxd \quad \leq  \sum_{\begin{subarray}{c}  j\in \co{m}{m+p+\diam}\\  x_j\neq y_j\end{subarray}}\sum_{i\in \oc{j-m-\diam }{j-m}} \maxd \\
&\le \sum_{\begin{subarray}{c}  j\in \co{m}{m+p+\diam}\\  x_j\neq y_j\end{subarray}} \diam  \maxd =  \dhh(x_{\co{m}{m+p+\diam}},y_{\co{m}{m+p+\diam}})\diam\maxd.
\qed\popQED
\end{align*}
\end{proof}

\begin{pro}\label{p:subbesi}
Let $F$ be a uniform dill map with diameter $\diam$ and local rule $f$. Then: 
$$\weyl{H}(F(x),F(y))\leq \frac{\diam \maxd}{\minf{f}}\cdot\weyl{H}(x,y), \forall x,y\in A^\N.$$
\end{pro}
\begin{proof}
Let us prove that $F$ is $\frac{\diam\maxd}{\minf{f}}$-Lipschitz with respect to $\weyl{H}$. According to Lemma \ref{lip:lemma}, for large enough $\ell$, for $k\in \N$, $m=\spart{\frac{k}{\minf f}}$ and $p=\ipart{\frac{\ell+k}{\minf f}}-(m+1)$ we obtain: 
$$\dhh(F(x)_{\co{k}{k+\ell}},F(y)_{\co{k}{k+\ell}})\le \dhh(x_{\co{m}{m+p}},y_{\co{m}{m+p}})\diam\maxd+\diam^2\maxd+\minf f.$$
Hence: 
\begin{eqnarray*}		
  \dfrac{\dhh(F(x)_{\co{k}{k+\ell}},F(y)_{\co{k}{k+\ell}})}{\ell} &\le &\dfrac{\max_{h\in\N}\dhh (x_{\co{h}{h+p}},y_{\co{h}{h+p}})\diam\maxd+\diam^2\maxd+2\minf f}{\ell}\\
  &\le  &\dfrac{\diam\maxd}{\len{f}} \cdot\dfrac{\max_{h\in\N}\dhh (x_{\co{h}{h+p}},y_{\co{h}{h+p}})}{p}+\dfrac{\diam^2\maxd+2\minf f}{\ell}.
\end{eqnarray*}
Since this was true for every $k$ and since $p\to \infty$ when $\ell\to \infty$, we obtain: 
$$ \weyl{H}(F(x),F(y))\leq \frac{\diam\maxd}{\minf{f}} \weyl{H}(x,y).\popQED$$
\end{proof}
\begin{pro}\label{uniorconst}
Let $F$ be a dill map with diameter $\diam\in\Ns$ and local rule $f$. 
If $F_{\weyl{H}}$ is well-defined, then $F$ is either constant or uniform.
\end{pro}
\begin{proof}
Assume that $F$ is non-uniform, \ie there are two words $u$ and $v$ of equal length such that $\len{f^*(u)}\ne\len{f^*(v)}$.
	One can assume that their longest common suffix has length $\diam-1$.
	Indeed, otherwise let $a\in A$, $u'=u_{\co{\len u-\diam+1}{\len u}}a^{\diam-1}$ and $v'=v_{\co{\len u-\diam+1}{\len v}}a^{\diam-1}$; one can note that $f^*(ua^{\diam-1})=f^*(u)f^*(u')$ and $f^*(va^{\diam-1})=f^*(v)f^*(v')$, so that either $\len{f^*(ua^{\diam-1})}\ne\len{f^*(va^{\diam-1})}$, or $\len{f^*(u')}=\len{f^*(ua^{\diam-1})}-\len{f^*(u)}\ne\len{f^*(va^{\diam-1})}-\len{f^*(v)}=\len{f^*(v')}.$ 
	Note that both of these pairs of words share a common suffix of length at least $\diam-1$.
	Assume also without loss of generality that $k=\len{f^*(u)}-\len{f^*(v)}>0$.
	\begin{itemize}
		\item First assume that there exist $w\in A^*$ and $i\in \N$ such that $f^*(w)_i\neq f^*(w)_{i+k}$.
		 By our previous assumption, we know that for $z=\uinf w$ and $w'=u_{\co{\len u-\diam}{\len u}}z_{\co0\diam}=v_{\co{\len v-\diam}{\len v}}z_{\co0\diam}$ we have: 
		 $F(uz)=f^*(u)f^*(w')F(z)$. 
		 According to the proof of \cite[Theorem 20]{ramdhane2022cellular}, we obtain: 
$\weyl{H}(F(uz),F(vz))\ge \besi{H}(F(uz),F(vz))\ge\frac1{|f^*(z_{\co0{\len w+\diam}})|}.$
		 Since $\len{u}=\len{v}$, $\weyl{H}(uz,vz)=0$, so that $F$ is not well-defined with respect to $\weyl{L}$.
		\item Otherwise, for all $w\in A^*$, $i\in \co{0}{\len{f^*(w)}-k}$, we have $f^*(w)_i=f^*(w)_{i+k}$.
              Then for every $x\in A^\N$, $F(x)$ is $k$-periodic and thus $F(x)=(f^*(x_{\co{0}{k+\diam}})_{\co{0}{k}})^\infty$.
                  Assuming that $F$ is well-defined, we get $\besi{H}(F(x),F(0^{k+\diam}x_{\co{k+\diam}\infty}))=0$. 
                According to Proposition \cite[Proposition 3]{cattaneo1997shift}, we can deduce that $F(x)=F(0^{k+\diam}x_{\co{k+\diam}\infty})$. 
                  Then, for every $x\in A^\N$, $F(x)
                  =\uinf{(f^*(0^{k+\diam})_{\co{0}{k}})}$.
		Hence, $F$ is constant.
			\popQED
\end{itemize}
\end{proof}
We now reach necessary and sufficient conditions for dill maps to induce well-defined maps over this space.
\begin{teo}\label{dill Besi}
Let $F$ be a dill map with diameter $\diam$ and local rule $f$. Then the following statements are equivalent: 
\begin{enumerate}
\item\label{i:def} $F_{\weyl{H}}$ is well-defined.
\item\label{i:lip} $F$ is $\frac{\diam\maxd}{\minf f}$-Lipschitz with respect to $\weyl{H}$. 
\item\label{i:cst} $F$ is either constant or uniform.
\end{enumerate}
\end{teo}
\begin{proof}~ \ref{i:lip}$\implies$\ref{i:def} is clear from the definition of Lipschitz maps.
  Implication \ref{i:cst}$\implies$\ref{i:lip} follows from Proposition~\ref{p:subbesi}.
  Implication \ref{i:def}$\implies$\ref{i:cst} follows from Proposition \ref{uniorconst}.
\end{proof}
 
\section{Lipschitz property of dill maps with respect to $\weyl{L}$}
In \cite{ramdhane2022cellular}, we proved that all dill maps are well-defined in the Feldman-Katok space. However, by changing the Feldman-Katok pseudo-metric for $\weyl{L}$, one can find that not all dill maps are well-defined. See for instance the following example:
\begin{ex}
Let $\tau$ be a substitution defined over $A=\{0,1\}^\N$ by $\tau(0)=0$ and $\tau(1)=11$.
Let $x=(0,1)^{(n,n)_{n\in\Ns}}$ and $y=(0,1)^{(n+1,n-1)_{n\in\Ns}}.$
Note that for all $j\in\N$, since $s_j:=j(j+1)=\sum_{i=1}^{j}2i$, we obtain: $d_H(x_{\co{s_j}{s_{j+1}}},y_{\co{s_j}{s_{j+1}}})=1.$
Let $\ell\in\Ns$ and $k\in\N$.
For $p=\min\sett{j\in\N}{s_j \ge k}$, $m=\max\sett{j\in\N}{s_j\le k+\ell}$ and by subadditivity:
\begin{eqnarray*}
d_H(x_{\co{k}{k+\ell}},y_{\co{k}{k+\ell}})&=& d_H(x_{\co{k}{s_p}},y_{\co{k}{s_p}})+d_H(x_{\co{s_p}{s_m}},y_{\co{s_p}{s_m}})+d_H(x_{\co{s_m}{k+\ell}},y_{\co{s_m}{k+\ell}}) \\
&\le & 1+d_H(x_{\co{s_p}{s_m}},y_{\co{s_p}{s_m}}) +1\le (m-p)+2.
\end{eqnarray*} 
Moreover, since $\ell+k\ge m^2+m$ and $k\le p^2+p$, we obtain $\frac{m-p}{\ell}\le \frac{1}{m+p+1}$, and thus: 
$$\dfrac{d_H(x_{\co{k}{k+\ell}},y_{\co{k}{k+\ell}})}{\ell}\le \dfrac{m-p+2}{\ell}\le \dfrac{2}{m+p+1}+\dfrac{2}{\ell}.$$
Since $m$ tends to $\infty$ when $\ell$ tends to $\infty$:
$\lim_{\ell\to\infty}\dfrac{d_H(x_{\co{k}{k+\ell}},y_{\co{k}{k+\ell}})}{\ell}=0, \forall k\in\N.$\\
Thus $\weyl{L}(x,y)=\weyl{H}(x,y)=0$.
On the other hand, it is clear that: $$\overline{\tau}(x)=(0,1)^{(n,2n)_{n\in\Ns}} \text{ and }\overline{\tau}(y)=(0,1)^{(n+1,2(n-1))_{n\in\Ns}}.$$
For $\ell\in\Ns$ and $k=\sum_{i=1}^{\ell}(i+1)+\sum_{i=1}^{\ell}2(i-1)$, hence: 
$\overline{\tau}(y)_{\co{k}{k+\ell}}=0^\ell.$
In contrast, since $\left(\sum_{i=1}^{\ell}i+\sum_{i=1}^{\ell}2i \right)-\ell=k$, we obtain $\overline{\tau}(x)_{\co{k}{k+\ell}}=1^\ell$. 
Thus: $$d_{L}(\overline{\tau}(x)_{\co{k}{k+\ell}},\overline{\tau}(y)_{\co{k}{k+\ell}})=d_L(0^\ell,1^\ell)=\ell.$$ 
Hence, $\max_{h\in\N}d_{L}(\overline{\tau}(x)_{\co{h}{h+\ell}},\overline{\tau}(y)_{\co{h}{h+\ell}})=\ell$. 
Therefore, $\weyl{L}(\overline{\tau}(x), \overline{\tau}(y))=1$.
\\
In conclusion, $\overline{\tau}$ is not well-defined with respect ot $\weyl{L}$.
\end{ex}
Before giving a caracterization of well-defined dill-maps with respect to $\weyl{L}$, let us recall a result from \cite{ramdhane2022cellular} to be used in the proof of the main result of this section. 
\begin{lem}[\cite{ramdhane2022cellular}]\label{l:Lev}
	Let $F$ be a dill map with diameter $\diam$ and local rule $f$.
	Then for all $\ell\in\N$ and $u,v\in A^\ell$, we have:
	$\dll(f^*(u),f^*(v))\leq\maxf f(2\diam-1)\dll(u,v)-\frac{\abs{\len{f^*(u)}-\len{f^*(v)}}}2.$
\end{lem}

Now let us characterize dill maps which induce a well-defined function over $X_{\weyl{L}}$.
\begin{defe}
We say that a dill map with diameter $\diam$ and local rule $f$ is \emph{diamond-uniform} if for every $\ell$ and $u,v\in A^\ell$ such that $u_{\co0\theta}=v_{\co0\theta}$ and $u_{\co{\ell-\theta}\ell}=v_{\co{\ell-\theta}\ell}$, one has $\len{f^*(u)}=\len{f^*(v)}$.
\end{defe}
It is clear that all uniform dill maps are diamond-uniform. Here is an example of non-uniform dill map which is diamond-uniform.
\begin{ex}
Let $F$ be the dill map with diameter $\diam=2$ and local rule $f$ defined by:
$$f(aa)=ab, f(bb)=ba, f(ab)=a \text{ and } f(ba)=bab.$$
It is clear that for any $x\in A^\N$, $F(x)=(ab)^\infty$ if $x$ start by the letter $a$ and $F(x)=(ba)^\infty$ otherwise.
And thus, $F$ is neither constant nor uniform.
However, $F$ is ${\weyl{L}}$-constant, since for every $x,y\in A^\N$, $\weyl{L}(F(x),F(y))=0$. So, $F_{\weyl{L}}$ is well-defined.
\end{ex}
%
%

\begin{teo}\label{t:carac}
If $F$ is a dill map, then $F_{\weyl{L}}$ is well-defined if and only if $F$ is $\weyl{L}$-constant or diamond-uniform.
\end{teo}
In order to prove this, here is the key point about diamond-uniformity.
\begin{lem}\label{c:di-uni}
If $F$ is a dill map with diameter $\diam$ and local rule $f$, then the following statements are equivalent:
  \begin{enumerate}
  \item\label{i:dunif} $F$ is diamond-uniform.
  \item\label{i:aunif} $\len{f^*(u)}-\len{f^*(v)}$ is uniformly bounded, for every $u,v$ with equal length.
  \end{enumerate}\end{lem}
\begin{proof}~
\begin{itemize}
  \item[$\ref{i:dunif}\Rightarrow\ref{i:aunif}$] Let $u,v$ have equal length.
    Then  
\begin{eqnarray*}    
    \len{f^*(u)}-\len{f^*(v)}&=&\len{f^*(a^{\theta-1}ua^{\theta-1})}-\len{f^*(a^{\theta-1}va^{\theta-1})}-\len{f^*(a^{\theta-1}u_{\co0{\theta-1}})}+ \\  &\quad +& \len{f^*(a^{\theta-1}v_{\co0{\theta-1}})}-\len{f^*(u_{\co{\len u-{\theta-1}}{\len u}}a^{\theta-1})}+\len{f^*(v_{\co{\len u-{\theta-1}}{\len u}}a^{\theta-1})}\\ &\le & 2(\theta-1)(\maxd-\mind). 
\end{eqnarray*}
  \item[$\ref{i:aunif}\Rightarrow\ref{i:dunif}$] Assume, by contrapositive, that there exist $\ell\in\N$ and $u,v\in A^\ell$ such that: $$u_{\co{0}{\diam-1}}=v_{\co{0}{\diam-1}}, u_{\co{\ell-\theta+1}\ell}=v_{\co{\ell-\theta+1}\ell} \text{ and }\len{f^*(u)}-\len{f^*(v)}\ne0.$$
      Then $\len{f^*(u^k)}-\len{f^*(v^k)}=(\len{f^*(u)}-\len{f^*(v)})k$, thanks to the common prefixes and suffixes.
      Hence $\len{f^*(u)}-\len{f^*(v)}$ is not bounded (and not upper-bounded, up to swapping $u$ and $v$).
\popQED\end{itemize}\end{proof}
\begin{pro}\label{p:subweyl}
Any diamond-uniform dill map $F$ with local rule $f$ and diameter $\diam$ is $(2\diam-1)\times\dfrac{\maxf f}{\minf f}$-Lipschitz with respect to $\weyl{L}$.
\end{pro}
\begin{proof}
Let $x,y\in A^\N$.
For large enough $\ell$ and $k\in\N$ let us denote: 
\begin{eqnarray*}
m &=&\max\left(\min\sett{i\in\N}{|f^*(x_{\co{0}{i+\diam}})|\ge k}, \min\sett{i\in\N}{|f^*(y_{\co{0}{i+\diam}})|\ge k}\right) \text{, and } \\  p &=&\min\left(\max\sett{i\in\N}{|f^*(x_{\co{0}{m+p}})|\le k+\ell}, \min\sett{i\in\N}{|f^*(y_{\co{0}{m+p}})|\le k+\ell}\right).
\end{eqnarray*}
Since $F$ is a diamond-uniform dill map and thanks to \ref{i:aunif}, there exist $C>0$ such that: 
$$\len{\len{f^*(x_{\co{0}{m}})}-\len{f^*(y_{\co{0}{m}})}}\le C \text{ and } \len{\len{f^*(x_{\co{0}{m+p}})}-\len{f^*(y_{\co{0}{m+p}})}}\le C.$$
And thus we can write, 
$$F(x)_{\co{k}{k+\ell}}=uvf^*(x_{\co{m}{m+p}})wz \text{ and } F(y)_{\co{k}{k+\ell}}=u'v'f^*(y_{\co{m}{m+p}})w'z'$$
where $|u|,|u'|,|z|,|z'|<\maxf f$, and $|v|,|v'|,|w|,|w'|<C$.
Hence, by subadditivity: 
\begin{eqnarray*}
d_L(F(x)_{\co{k}{k+\ell}},F(y)_{\co{k}{k+\ell}})&\le & d_L(uv,u'v')+d_L(f^*(x_{\co{m}{m+p}}),f^*(y_{\co{m}{m+p}}))+d_L(wz,w'z') \\
&\le & 2(\maxf f+C)+d_L(f^*(x_{\co{m}{m+p}}),f^*(y_{\co{m}{m+p}})).
\end{eqnarray*}
According to Lemma \ref{l:Lev} we obtain:
$$d_L(F(x)_{\co{k}{k+\ell}},F(y)_{\co{k}{k+\ell}})\le 2(\maxf f+C)+(2\diam-1)\maxf f d_L(x_{\co{m}{m+p}},y_{\co{m}{m+p}}).$$
Since $\ell\ge  \len{f^*(x_{\co{m}{m+p}})}\ge (p-\diam)\minf f$ and by subadditivity: 
\begin{eqnarray*}
\dfrac{d_L(F(x)_{\co{k}{k+\ell}},F(y)_{\co{k}{k+\ell}})}\ell &\le & \dfrac{\maxf f(2+\diam(2\diam-1))+2C}{\ell}+\dfrac{(2\diam-1)\maxf f d_L(x_{\co{m}{m+p-\diam}},y_{\co{m}{m+p-\diam}})}{\minf f(p-\diam)}\\
&\le & \dfrac{\maxf f(2\diam^2-\diam+2)+2C}{\ell}+(2\diam-1)\frac{\maxf f}{\minf f}\times \max_{h\in \N}\dfrac{d_L(x_{\co{h}{h+p-\diam}},y_{\co{h}{h+p-\diam}})}{p-\diam}.
\end{eqnarray*}
Since this was true for every $k\in\N$ and since $p\to \infty$ when $\ell\to\infty$: 
$$\weyl{L}(F(x),F(y))\le (2\diam-1)\dfrac{\maxf f}{\minf f}\weyl{L}(x,y).\popQED$$
\end{proof}
\begin{proof}[Proof of Theorem~\ref{t:carac}]~According to Proposition \ref{p:subweyl}, if $F$ is diamond-uniform then $F_{\weyl{L}}$ is well-defined.
Suppose now that $F$ is neither  $\weyl{L}$-constant nor diamond-uniform \ie there exists $x,y\in A^\N$ such that $\weyl{L}(F(x),F(y))>0$, and there exists $m\in \N$, $u,v\in A^m$ such that $u_{\co{0}{\diam-1}}=v_{\co{0}{\diam-1}}$, $u_{\co{m-\diam+1}{m}}=v_{\co{m-\diam+1}{m}}$ and $\alpha\defeq\len{f^*(u)}-\len{f^*(v)}>0$.
  Let us define the following configurations:
  \begin{eqnarray*}
    z&\defeq & ux_{\co0\alpha}y_{\co0\alpha}ux_{\co0{2\alpha}}y_{\co0{2\alpha}}ux_{\co0{3\alpha}}y_{\co0{3\alpha}}\ldots, \text{ and } \\ w &\defeq & vx_{\co0\alpha}y_{\co0\alpha}vx_{\co0{2\alpha}}y_{\co0{2\alpha}}vx_{\co0{3\alpha}}y_{\co0{3\alpha}}\ldots. 
  \end{eqnarray*}
  Let $(k_\ell)_{\ell\in\N}$ such that for all $\ell\in\N$: $\max_{k\in\N}d_L(F(x)_{\co{k}{k+\ell}},F(y)_{\co{k}{k+\ell}})=d_L(F(x)_{\co{k_\ell}{k_\ell+\ell}},F(y)_{\co{k_\ell}{k_\ell+\ell}}).$
  This pair of patterns appears in $(z,w)$, in the zone where $n\alpha\ge k_\ell+\ell$.
  Hence $\weyl{L}(z,w)=0$ but $\weyl{L}(F(z),F(w))\ge \weyl{L}(F(x),F(y))>0$.
\end{proof}
\section{Conclusion and perspectives}
In this paper, we characterized well-defined dill maps over the Weyl space, indeed, we find the same result as in the case of Besicovitch space in \cite{ramdhane2022cellular}. In addition, we showed that not all dill maps are well-defined with respect to the sliding Feldman-Katok space, in contrast the Feldman-Katok space, where all dill maps are well-defined \cite[Corollary 46]{ramdhane2022cellular}.
Those spaces were constructed using two pseudo-metrics depending on two different edit distances over finite words (the Hamming distance and the Levenshtein distance). One natural question is which properties on distance $d$ make all dill maps are well-defined in the corresponding pseudo-metric space?
\bibliographystyle{alpha}
\bibliography{besiweyl}
\end{document}

Firas
⊕posexp : pér Feldman. maximum. shifted=max centered
- equivalent with Besicovitch
- distance sur les mots de même longueur (Ex 3.24)
Automata :
-both pseudometrics as natural additive ones
-behavior sliding pseudometrics

-open question positive expansiveness of some CA/dill map?